\documentclass[11pt]{article}

\usepackage{fullpage}
\usepackage{amsmath,amsthm,amsfonts,dsfont}
\usepackage{amssymb,latexsym,graphicx}
\usepackage{mathtools}
\usepackage{hyperref}
\usepackage{xcolor}
\hypersetup{
	colorlinks,
	linkcolor={red!75!black},
	citecolor={blue!75!black},
	urlcolor={blue!75!black}
}

\newtheorem{theorem}{Theorem}[section]

\newtheorem{lemma}[theorem]{Lemma}

\newtheorem{claim}[theorem]{Claim}

\newtheorem{corollary}[theorem]{Corollary}

\newcommand{\R}{\ensuremath{\mathbb{R}}}
\newcommand{\Z}{\ensuremath{\mathbb{Z}}}

\newcommand{\lat}{\mathcal{L}}

\newcommand{\eps}{\varepsilon} 
\renewcommand{\epsilon}{\varepsilon}
\newcommand{\poly}{\mathrm{poly}}

\DeclareMathOperator*{\expect}{\mathbb{E}}
\DeclareMathOperator{\dist}{dist}

\renewcommand{\vec}[1]{\ensuremath{\boldsymbol{#1}}}

\newcommand{\gs}[1]{\ensuremath{\widetilde{\vec{#1}}}}

\DeclarePairedDelimiter\inner{\langle}{\rangle}

\mathtoolsset{centercolon}
\usepackage[hyperpageref]{backref}

\begin{document}
	\title{A time-distance trade-off for GDD with preprocessing---\\Instantiating the DLW heuristic}
	\author{
		Noah Stephens-Davidowitz\\Massachusetts Institute of Technology\\
		\texttt{noahsd@gmail.com}
	}
	\date{}
	\maketitle

\begin{abstract}
	For $0 \leq \alpha \leq 1/2$, we show an algorithm that does the following. Given appropriate preprocessing $P(\lat)$ consisting of $N_\alpha := 2^{O(n^{1-2\alpha} + \log n)}$ vectors in some lattice $\lat \subset \R^n$ and a target vector $\vec{t}\in \R^n$, the algorithm finds $\vec{y} \in \lat$ such that $\|\vec{y}- \vec{t}\| \leq n^{1/2 + \alpha} \eta(\lat)$ in time $\poly(n) \cdot N_\alpha$, where $\eta(\lat)$ is the smoothing parameter of the lattice.
	
	The algorithm itself is very simple and was originally studied by Doulgerakis, Laarhoven, and de Weger (to appear in PQCrypto, 2019), who proved its correctness under certain reasonable heuristic assumptions on the preprocessing $P(\lat)$ and target $\vec{t}$. Our primary contribution is a choice of preprocessing that allows us to prove correctness without any heuristic assumptions.
	
	Our main motivation for studying this is the recent breakthrough algorithm for IdealSVP due to Hanrot, Pellet--Mary, and Stehl\'e (to appear in Eurocrypt, 2019), which uses the DLW algorithm as a key subprocedure. In particular, our result implies that the HPS IdealSVP algorithm can be made to work with fewer heuristic assumptions.
	
	Our only technical tool is the discrete Gaussian distribution over $\lat$, and in particular, a lemma showing that the one-dimensional projections of this distribution behave very similarly to the continuous Gaussian. This lemma might be of independent interest.
\end{abstract}

\section{Introduction}

A lattice $\lat \subset \R^n$ is the set of all integer linear combinations 
\[
\lat := \{z_1 \vec{b}_1 + \cdots + z_n \vec{b}_n  \ : \ z_i \in \Z \}
\]
of linearly independent basis vectors $\vec{b}_1,\ldots, \vec{b}_n \in \R^n$. For a lattice $\lat \subset \R^n$ and target vector $\vec{t}\in \R^n$, the $d$-Guaranteed Distance Decoding problem ($d$-GDD, or just GDD) asks us to find $\vec{y}\in \lat$ such that $\|\vec{y}- \vec{t}\| \leq d$ for some distance $d := d(\lat)$ that depends only on $\lat$. In particular, we must have $d \geq \mu(\lat)$, where $\mu(\lat) := \max \dist(\vec{t}, \lat)$ is the covering radius of the lattice.

GDD with preprocessing (GDDP) is the variant of this problem in which we are allowed to perform arbitrary preprocessing on the lattice (but not on $\vec{t}$). I.e., formally an ``algorithm'' for GDDP is really a \emph{pair} of algorithms, a preprocessing algorithm, which takes as input (a basis for) a lattice $\lat \subset \R^n$ and outputs some preprocessing $P(\lat)$, and a query algorithm which takes as input $P(\lat)$ and a target $\vec{t}$ and outputs a valid solution to the GDD instance $(\lat, \vec{t})$. The complexity measure that interests us for such algorithms is the running time of the query algorithm.

In~\cite{DLdFindingClosest19}, Doulgerakis, Laarhoven, and de Weger (DLW) gave an elegant algorithm for GDDP whose correctness relies on certain heuristic assumptions.
(Our presentation here differs quite a bit from DLW's. See Section~\ref{sec:laarhoven_related}.) 
In fact,~\cite{DLdFindingClosest19} gave a family of algorithms parameterized by $0 \leq \alpha \leq 1/2$ whose preprocessing consists of $N_\alpha \approx 2^{n^{1-2\alpha}}$ lattice vectors in $\lat$ whose length is roughly $r$. Given a target $\vec{t}$, the query algorithm starts by setting $\vec{t}' = \vec{t}$. The algorithm then simply searches for a vector $\vec{y}$ in the preprocessing list and an integer $k$ such that $\|k\vec{y}- \vec{t}'\| < \|\vec{t}'\|$. If it finds one, it replaces $\vec{t}'$ by $\vec{t}'- k\vec{y}$ and repeats this procedure. Finally, it outputs $\vec{y}' := \vec{t} - \vec{t}' \in \lat$. Under certain heuristic assumptions that in particular imply that the preprocessing is nicely distributed,~\cite{DLdFindingClosest19} showed that this algorithm terminates with $\|\vec{y}' - \vec{t}\| = \|\vec{t}'\| \lesssim n^\alpha \cdot r$ in time $\poly(n) \cdot N_\alpha$.

DLW's algorithm is the first to provide a smooth trade-off between the running time and the distance $d$. (Such trade-offs are known for other lattice problems. E.g., without preprocessing, block reduction~\cite{SchHierarchyPolynomial87,GNFindingShort08} algorithms accomplish this for many lattice problems, and with preprocessing, such trade-offs are known for Bounded Distance Decoding and the Closest Vector Problem~\cite{LLMBoundedDistance06,DRSClosestVector14}.) This recently found an exciting application discovered by Pellet--Mary, Hanrot, and Stehl\'e~\cite{PHS19}.~\cite{PHS19} showed the best known time-approximation-factor trade-off for the very important problem of finding short non-zero vectors in ideal lattices (given suitable preprocessing on the underlying number field). Their algorithm uses the DLW algorithm as a key subprocedure. However, since DLW's algorithm relies on certain heuristic assumptions, their application crucially relies on the (reasonable but unproven) assumption that these heuristics apply in their particular use case.

\subsection{Removing the heuristic in DLW's GDDP algorithm}

We show how to instantiate DLW's heuristic algorithm in a provably correct way. In particular, we show an explicit distribution over the lattice such that, when the preprocessing consists of independent samples from this distribution, the above algorithm provably succeeds with high probability. Indeed, there is a very natural choice for this distribution: the discrete Gaussian over the lattice, $D_{\lat, s}$. This is the distribution that assigns probability to each lattice vector $\vec{y}\in \lat$ proportional to its Gaussian mass $\exp(-\pi \|\vec{y}\|^2/s^2)$, and it is a ubiquitous tool in lattice algorithms and the study of lattices more generally. (See, e.g.,~\cite{NSDthesis}.) When the width parameter $s > 0$ is at least as large as the \emph{smoothing parameter} $\eta(\lat)$, the discrete Gaussian distribution $D_{\lat, s}$ provably behaves quite similarly to the continuous Gaussian in many ways~\cite{MRWorstcaseAveragecase07}. (E.g., its moments are close to those of a continuous Gaussian.) So, one might expect that it will be distributed nicely enough to work for DLW's use case.

We show that for $s = \eta(\lat)$, the discrete Gaussian $D_{\lat, s}$ does in fact suffice to provably instantiate DLW's heuristic algorithm with $r \approx \sqrt{n} \cdot \eta(\lat)$. (This is essentially the same value of $r$ used in~\cite{DLdFindingClosest19}. See Section~\ref{sec:laarhoven_related} for more discussion.) I.e., we prove the following theorem.

\begin{theorem}
	\label{thm:intro}
	For any $0 \leq \alpha \leq 1/2$, there is an algorithm that solves $d$-GDDP in time $2^{O(n^{1-2\alpha} + \log n)}$ where $d(\lat) := n^{1/2 + \alpha} \cdot \eta(\lat)$.
\end{theorem}

Theorem~\ref{thm:intro} is primarily interesting for $\alpha$ strictly between zero and $1/2$. For $\alpha = 0$, Theorem~\ref{thm:intro} is outperformed by existing $2^{O(n)}$-time algorithms for CVP~\cite{MVDeterministicSingle13,ADSSolvingClosest15}. These algorithms do not require preprocessing and are actually guaranteed to find a \emph{closest} vector to the target $\vec{t}$, so our algorithm is trounced by the competition in this regime. Similarly, for $\alpha = 1/2$, Babai's celebrated polynomial-time algorithm~\cite{BabLovaszLattice86} always matches or outperforms Theorem~\ref{thm:intro} when instantiated with an appropriate basis as preprocessing.

However, for parameters $\alpha$ satisfying $C/\log n \leq \alpha \leq 1/2 - C/\log n$ for a sufficiently large constant $C > 0$, Theorem~\ref{thm:intro} is the best known algorithm and the first non-trivial result whose correctness has been proven.%
\footnote{Formally, Babai's algorithm can be used to solve $d$-GDDP in polynomial time for a function $d(\lat)$ satisfying $\eta(\lat) \lesssim d(\lat) \lesssim n \cdot \eta(\lat)$. (This function is given by $d(\lat)^2 = \min \sum \|\gs{\vec{b}}_i\|^2/4$, where the minimum is over all (ordered) lattice bases $\vec{b}_1,\ldots, \vec{b}_n$ and $\gs{\vec{b}}_i$ represents the Gram-Schmidt orthogonalization.) So, the distance $n^{1/2 + \alpha} \eta(\lat)$ that we achieve is incomparable with Babai's for $\alpha < 1/2$. Indeed, for some rather degenerate lattices---such as the lattice generated by $2^n \vec{e}_1, \vec{e}_2,\ldots, \vec{e}_n$---we have $d(\lat) \approx \eta(\lat)$. I.e., Babai's algorithm outperforms Theorem~\ref{thm:intro} by a factor of roughly $n^{1/2 + \alpha}$ in the distance for such lattices. We can of course always combine the two algorithms to achieve the best of both worlds, a distance of $\min\{d(\lat) , n^{1/2 + \alpha} \eta(\lat) \}$. However, for ``typical'' lattices that interest us, like those that satisfy the heuristics in~\cite{DLdFindingClosest19} or those that arise in cryptography, we have $d(\lat) \approx n \cdot \eta(\lat)$, so that our algorithm strictly outperforms Babai's for $\alpha \leq 1/2 - C/\log n$.}
In particular, Theorem~\ref{thm:intro} removes~\cite{PHS19}'s reliance on certain heuristic assumptions.  (\cite{PHS19} also requires additional unrelated heuristic assumptions, which our result does not remove. We refer the reader to~\cite{PHS19} for more information.)

Behind this result is a geometric lemma concerning the discrete Gaussian distribution that, to the author's knowledge, is novel. The lemma shows that one-dimensional projections of the discrete Gaussian look very much like a continuous Gaussian for parameters above smoothing. (See Lemma~\ref{lem:DGS_one_dimensional_projection}.)

\subsection{Relation to DLW}
\label{sec:laarhoven_related}

Our presentation here is quite different from the presentation in~\cite{DLdFindingClosest19}. (See also an earlier version of the same paper~\cite{LaaFindingClosest16} and a closely related paper~\cite{LaaSievingClosest16}.) We attempt to clarify some of the differences here to avoid confusion.

First of all, DLW described their algorithm as a solution to the \emph{Closest Vector Problem} (CVP), in which the goal is to output a vector $\vec{y}\in \lat$ with $\|\vec{y} - \vec{t}\| \leq \gamma \cdot \dist(\vec{t}, \lat)$ for some approximation factor $\gamma \geq 1$. In contrast, we call the same algorithm a GDD(P) algorithm. This discrepancy arises when one moves from heuristic algorithms to provably correct algorithms. Since $\dist(\vec{t}, \lat)$ is nearly maximal for ``most'' $\vec{t}$~\cite{HLRNoteDistribution09}, DLW's heuristics quite reasonably imply that $\dist(\vec{t}, \lat)$ is nearly maximal, i.e., $\dist(\vec{t}, \lat) \approx \mu(\lat)$. With this assumption, $\gamma$-CVP is essentially equivalent to $(\gamma \mu(\lat))$-GDD. However, without such a heuristic, the two problems seem to be quite different, so that the distinction is unfortunately necessary here. 

Second, since~\cite{DLdFindingClosest19} describe their results in terms of CVP and do not mention the smoothing parameter $\eta(\lat)$, their results are formally incomparable with Theorem~\ref{thm:intro}. However, we note that the heuristics in~\cite{DLdFindingClosest19} imply that $\eta(\lat) \approx \lambda_1(\lat)/\sqrt{n} \approx \mu(\lat)/\sqrt{n}$, and the DLW algorithm finds vectors within distance roughly $n^{\alpha} \lambda_1(\lat) $ of the target. Since we obtain vectors within distance $n^{1/2+\alpha} \eta(\lat)$, our result essentially matches theirs when their heuristics apply.

Third, while we match DLW's algorithm asymptotically, we do not claim to match the constants. Indeed, in the language of this paper, much of~\cite{DLdFindingClosest19} is devoted to finding vectors within distance $c_1 \sqrt{n} \cdot \eta(\lat)$ in time $2^{c_2 n + o(n)}$ for small  constants $0 < c_1, c_2 < 1$. In contrast, we are mostly interested in what appears as a secondary result in that paper: the time-distance trade-off achievable for distance $n^{1/2 + \alpha} \eta(\lat)$ and time $2^{O(n^{1-2\alpha} + \log n)}$ for $0 < \alpha < 1/2$. And, we make very little effort to optimize the constants. For example,~\cite{DLdFindingClosest19} uses nearest neighbor data structures to let the query algorithm avoid reading the entire preprocessing, which we do not attempt to replicate here. Similarly, while~\cite{DLdFindingClosest19} proposed specific techniques for computing the preprocessing in $2^{c n + o(n)}$ time, we ignore this. (We do note, however, that~\cite{ADRSSolvingShortest15} shows how to sample the preprocessing in time $2^{n + o(n)}$.)

\subsection*{Acknowledgments}
I thank Guillaume Hanrot, Thijs Laarhoven, Alice Pellet--Mary, Oded Regev, and Damien Stehlé for helpful discussions. I also thank Alice Pellet--Mary, Guillaume Hanrot, and Damien Stehlé for sharing early versions of their work with me.

\section{Preliminaries}

Throughout this work, we adopt the common convention of expressing the running times of lattice algorithms in terms of the dimension $n$ only, ignoring any dependence on the bit length of the input $B$. Formally, we should specify a particular input format for the (basis of the) lattice (e.g., by restricting our attention to rational numbers and using the natural binary representation of a rational matrix to represent a basis for the lattice), and our running time should of course have some dependence on $B$. Consideration of the bit length would simply add a $\poly(B)$ factor to the running time for the algorithm(s) considered in this paper, provided that the input format allows for efficient arithmetic operations.

\subsection{The discrete Gaussian}

For a vector $\vec{x} \in \R^n$ and parameter $s > 0$, we write $\rho_s(\vec{x}) := \exp(-\pi \|\vec{x}\|^2/s^2)$ for the Gaussian mass of $\vec{x}$ with parameter $s$. For a lattice $\lat \subset \R^n$ and shift vector $\vec{t}\in \R^n$, we write
\[
\rho_s(\lat - \vec{t}) := \sum_{\vec{y}\in \lat} \rho_s(\vec{y} - \vec{t})
\]
for the Gaussian mass of $\lat - \vec{t}$ with parameter $s$. We write $D_{\lat, s}$ for the probability distribution over $\lat$ defined by
\[
\Pr_{\vec{X} \sim D_{\lat, s}}[\vec{X}= \vec{y}] = \frac{\rho_s(\vec{y})}{\rho_s(\lat)}
\]
for $\vec{y} \in \lat$.

The dual lattice $\lat^* \subset \R^n$ is the set of vectors that have integer inner product with all lattice vectors,
\[
\lat^* := \{\vec{w} \in \R^n \ : \ \forall \vec{y}\in \lat,\ \inner{\vec{w}, \vec{y}} \in \Z \}
\; .
\]
Micciancio and Regev defined the \emph{smoothing parameter} $\eta(\lat)$ as the unique parameter $s$ such that $\rho_{1/s}(\lat^*) = 3/2$~\cite{MRWorstcaseAveragecase07}.\footnote{This is more commonly referred to as $\eta_{1/2}(\lat)$, where $\eta_\eps(\lat)$ is the unique parameter $s$ such that $\rho_{1/s}(\lat^*) = 1+\eps$. Since we will always take $\eps = 1/2$, we simply omit it. Our results remain essentially unchanged if we take $\eps$ to be any constant strictly between zero and one.} The following claim justifies the name ``smoothing parameter,'' and it is the only fact about the smoothing parameter that we will need.

\begin{claim}
	\label{clm:smooth}
	For any lattice $\lat \subset \R^n$, parameter $s \geq \eta(\lat)$, and shift $\vec{t} \in \R^n$,
	\[
	\frac{1}{3}  \leq \frac{\rho_s(\lat - \vec{t})}{\rho_s(\lat)} \leq 1
	\; .
	\]
\end{claim}

We will also need a simplified version of Banaszczyk's celebrated tail bound for the discrete Gaussian~\cite{BanNewBounds93}.

\begin{theorem}
	\label{thm:lazy_banaszczyk}
	For any lattice $\lat \subset \R^n$ and parameter $s > 0$,
		\[
		\Pr_{\vec{X} \sim D_{\lat, s}}[\|\vec{X}\| \geq \sqrt{n} s] \leq 2^{-n}
		\; .
		\]
\end{theorem}

Finally, we will need the following rather weak consequence of Babai's algorithm~\cite{BabLovaszLattice86}.

\begin{theorem}
	\label{thm:babai}
	There is a polynomial-time algorithm for $(2^n \eta(\lat))$-GDD.
\end{theorem}

\subsection{\texorpdfstring{$\eps$-nets}{Eps-nets}}

	For $\eps > 0$, we say that a set $\{\vec{v}_1,\ldots, \vec{v}_M\} \subset \R^n$ of unit vectors with $\|\vec{v}_i\| = 1$ is an $\eps$-net of the unit sphere if for any $\vec{t} \in \R^n$ with $\|\vec{t}\| = 1$, there exists $\vec{v}_i$ such that $\|\vec{v}_i - \vec{t}\| \leq \eps$. We will use a simple bound on the size of such a net, which can be proven via a simple packing argument. See \cite[Lemma 5.2]{VerIntroductionNonasymptotic12}, for example.

\begin{lemma}
	\label{lem:smallnet}
	For any $\eps > 0$, there exists an $\eps$-net of the unit sphere in $\R^n$ with $(1+2/\eps)^n$ points.
\end{lemma}

\section{The algorithm}

We consider the following algorithm for GDDP. For an input lattice $\lat \subset \R^n$ with $n \geq 40$, the preprocessing consists of $N$ lattice vectors $\vec{y}_1,\ldots, \vec{y}_N \in \lat$. On input $\vec{t} \in \R^n$, the query algorithm behaves as follows. It first uses Theorem~\ref{thm:babai} to find $\vec{t}_0 \in \lat + \vec{t}$ such that $\|\vec{t}_0\| \leq 2^n \eta(\lat)$ and sets $j = 0$. The algorithm then does the following repeatedly. It finds an index $i$ and integer $k$ such that $\|\vec{t}_{j} - k\vec{y}_i\|^2 \leq (1-1/n^2) \cdot \|\vec{t}_{j}\|^2$, sets $\vec{t}_{j+1} := \vec{t}_{j} - k\vec{y}_i$, and increments $j$. Once the algorithm fails to find such a vector, it outputs $\vec{t}_{j} - \vec{t} \in \lat$.\footnote{To guarantee a running time of $\poly(n) \cdot N$, we can also assume that the algorithm halts and outputs $\vec{t}_j - \vec{t} \in \lat$ if $j$ reaches, say, $100n^3$. This is not strictly necessary, since we will have $\|\vec{y}_i\| \approx \sqrt{n} \cdot \eta(\lat)$ with very high probability.}

Our main theorem shows that this algorithm will succeed when the preprocessing is chosen from the right distribution.
We emphasize the order of quantifiers: with high probability over the preprocessing, the algorithm works \emph{for all targets $\vec{t} \in \R^n$}. In particular, there exists fixed preprocessing that works for all targets $\vec{t}$.

\begin{theorem}
	\label{thm:main_result}
	For any $\alpha$ with $\frac{2}{\log n} \leq \alpha \leq \frac1 2$, when the preprocessing of the above algorithm consists of $N_\alpha := n^2 e^{(n^{1/2-\alpha} + 4)^2} = 2^{O(n^{1-2\alpha} + \log n)}$ samples from $D_{\lat,s}$ for $s := \eta(\lat)$, it yields a solution to $d$-GDDP in time $\poly(n) \cdot N_\alpha$ with high probability, where $d := n^{1/2 + \alpha} \cdot \eta(\lat)$.
	\end{theorem}
\begin{proof}
	By scaling appropriately, we may assume without loss of generality that $d = 1$, and therefore $s = n^{-1/2 - \alpha}$. Let $\vec{y}_1,\ldots, \vec{y}_{N_\alpha} \sim D_{\lat,s}$. To prove correctness, we must show that, with high probability over the $\vec{y}_i$, for every $\vec{t} \in \R^n$ with $\|\vec{t}\| \geq 1$, there exists an index $i$ and integer $k$ such that $\|\vec{t} - k\vec{y}_i\|^2 \leq (1-1/n^2) \cdot \|\vec{t}\|^2$. It suffices to prove that for $\|\vec{t}\| = 1$, there exists an $i$ with $\|\vec{t}- \vec{y}_i\|^2 \leq 1-4/n^2$.%
	\footnote{Indeed, suppose that $\|\vec{t} - \vec{y}_i\|^2 \leq 1-4/n^2$ and $\|\vec{t}\| = 1$, so that in particular $\inner{\vec{y}_i, \vec{t}} \geq 0$. Then for any $\beta/2 \leq k \leq \beta$, 
		\[
		\frac{\|\beta \vec{t} - k\vec{y}_i\|^2}{ \|\beta \vec{t}\|^2} = 1 - \frac{k^2}{\beta^2} \cdot \big(1- \|\vec{t} - \vec{y}_i\|^2\big) - \Big(\frac{2k}{\beta} - \frac{2k^2}{\beta^2}\Big) \cdot\inner{\vec{y}_i, \vec{t}} \leq 1 - \frac{k^2}{\beta^2} \cdot \frac{4}{n^2} \leq 1 - \frac{1}{n^2}
		\; .
	\]
}
Finally, to prove \emph{this}, it suffices to take a $(1/n^3)$-net of the unit sphere, $\vec{v}_1,\ldots, \vec{v}_M$, and to show that for each $j$, there exists an $i$ such that $\|\vec{v}_j - \vec{y}_i\|^2 \leq 1-5/n^2$.

By Lemma~\ref{lem:smallnet}, there exists such a net of cardinality $M = (3n)^{3n}$. For each $\vec{v}_j$ in this net and each index $i$, we have by Corollary~\ref{cor:get_shorter} (proven below) that
	\begin{align*}
		\Pr\big[\|\vec{v}_j - \vec{y}_{i}\|^2 \leq 1-5/n^2\big] 
				&\geq \exp(-\pi (5/(n^2s)+ns + 4)^2/4) -  2^{-n}\\
				&\geq \exp(- (n^{1/2-\alpha} + 4)^2)\\
				&= n^2/N_\alpha
				\; .
	\end{align*}
	Since the $\vec{y}_{i}$ are sampled independently, the probability that no such $i$ exists is at most
	$
	(1-n^2/N_\alpha)^{N_\alpha} < 2^{-n}/M$. The result then follows by taking a union bound over the $\vec{v}_j$.
\end{proof}

\subsection{One-dimensional projections of the discrete Gaussian}

We are interested in the lower bound in the following lemma (whose proof uses an idea from~\cite{MOLatticePoints90}). The upper bound (i.e., the subgaussianity of the discrete Gaussian) applies for all parameters $s > 0$ and was first proven in~\cite[Lemma 2.8]{MPTrapdoorsLattices12}. We include the upper bound for comparison.

\begin{lemma}
	\label{lem:DGS_one_dimensional_projection}
	For any lattice $\lat \subset \R^n$, parameter $s \geq \eta(\lat)$, unit vector $\vec{v} \in \R^n$ with $\|\vec{v}\| = 1$, and $r_0 > 0$, we have
	\begin{equation}
	\label{eq:sub_supergaussian}
		\exp(-\pi (r_0/s+2)^2) < \Pr_{\vec{X} \sim D_{\lat, s}}\big[\inner{\vec{X}, \vec{v}} \geq r_0 \big] \leq \exp(-\pi r_0^2/s^2)
		\; .
	\end{equation}
\end{lemma}
\begin{proof}
	By scaling appropriately, we may assume that $s = 1$. Let $\beta > 0$ to be chosen later. By completing the square in the exponent, we see that
\[
		\expect[\exp(2\pi \beta \inner{\vec{X}, \vec{v}})] = \exp(\pi \beta^2) \cdot \frac{\rho_1(\lat - \beta \vec{v})}{\rho_1(\lat)} 
\; .
\]
	Therefore, by Claim~\ref{clm:smooth},
	\begin{equation}
	\label{eq:smooth_beta}
		\frac{1}{3} \leq \exp(-\pi \beta^2 ) \cdot \expect[\exp(2\pi \beta \inner{\vec{X}, \vec{v}})] \leq 1
		\; .
	\end{equation}
	I.e., we know the moment generating function of $\inner{\vec{X}, \vec{v}}$ to within a multiplicative constant.
	The upper bound in Eq.~\eqref{eq:sub_supergaussian} then follows from taking $\beta = r_0$ and applying Markov's inequality. (This proof of the upper bound is identical to the proof in~\cite{MPTrapdoorsLattices12}. See their Lemma 2.8 and their discussion above it.)

	Turning to the lower bound, for $r \in \R$, let
	\[
	f(r) := \exp(2\pi \beta r) \cdot \big(1- \exp(2\pi(r_0 - r)) - \exp(2\pi(r-r_0 - 2)) \big) 
	\; .
	\]
	Notice that $f(r) < 0$ unless $r_0 < r < r_0+2$. And, $f(r) < \exp(2\pi \beta (r_0 + 2))$ for all $r$.
	Therefore,
	\begin{equation}
	\label{eq:f_upper}
	\expect\big[f(\inner{\vec{X}, \vec{v}})\big] < \exp(2\pi \beta (r_0 + 2)) \cdot \Pr\big[\inner{\vec{X}, \vec{v}} \geq r_0\big]
	\; .
	\end{equation}
	By applying Eq.~\eqref{eq:smooth_beta} term-wise and taking $\beta = r_0+1$, we have
	\begin{align}
		\expect\big[f(\inner{\vec{X}, \vec{v}})\big] 
			&\geq \frac{1}{3} \cdot \exp(\pi \beta^2) - \exp(\pi (\beta - 1)^2 + 2\pi  r_0) - \exp(\pi (\beta + 1)^2 - 2\pi (r_0+2)) \nonumber \\
			&= \exp(\pi r_0^2 + 2\pi r_0) \cdot (e^\pi/3- 2  ) \nonumber \\
			&> \exp( \pi r_0^2 + 2\pi r_0) \label{eq:f_lower}
		\; .
	\end{align}
	Combining Eqs.~\eqref{eq:f_upper} and~\eqref{eq:f_lower} and rearranging, we have
	\[
		\Pr\big[\inner{\vec{X}, \vec{v}} \geq r_0\big] > \exp(\pi r_0^2 + 2\pi r_0-2\pi \beta (r_0 + 2)) = \exp(-\pi (r_0+2)^2)
		\; ,
	\]
	as needed.
\end{proof}

\begin{corollary}
	\label{cor:get_shorter}
	For any $0 < r < 1$, lattice $\lat \subset \R^n$, unit vector $\vec{v} \in \R^n$ with $\|\vec{v}\| = 1$, and $s \geq \eta(\lat)$, we have
	\[
	\Pr_{\vec{X} \sim D_{\lat, s}}\big[\|\vec{v} - \vec{X}\|^2 \leq 1 - r \big] > \exp(-\pi (r/s+ns + 4)^2/4) -  2^{-n}
	\; .
	\]
\end{corollary}
\begin{proof}
	Notice that $\|\vec{v} - \vec{X}\|^2 = 1+\|\vec{X}\|^2-2\inner{\vec{v}, \vec{X}}$. By Theorem~\ref{thm:lazy_banaszczyk}, we have that $\|\vec{X}\| \leq n s^2$ except with probability at most $2^{-n}$. By Lemma~\ref{lem:DGS_one_dimensional_projection}, we see that
	\[
	\Pr[\inner{\vec{v}, \vec{X}} \geq (n s^2 + r)/2] > \exp(-\pi (r/s+ns + 4)^2/4)
	\; .
	\]
	The result follows from union bound.
\end{proof}

\bibliographystyle{alpha}

\begin{thebibliography}{DLdW19}

\bibitem[ADRS15]{ADRSSolvingShortest15}
Divesh Aggarwal, Daniel Dadush, Oded Regev, and Noah {Stephens-Davidowitz}.
\newblock Solving the {Shortest Vector Problem} in $2^n$ time via {Discrete
  Gaussian Sampling}.
\newblock In {\em STOC}, 2015.
\newblock \url{http://arxiv.org/abs/1412.7994}.

\bibitem[ADS15]{ADSSolvingClosest15}
Divesh Aggarwal, Daniel Dadush, and Noah {Stephens-Davidowitz}.
\newblock Solving the {Closest Vector Problem} in $2^n$ time--{The} discrete
  {Gaussian} strikes again!
\newblock In {\em FOCS}, 2015.
\newblock \url{http://arxiv.org/abs/1504.01995}.

\bibitem[Bab86]{BabLovaszLattice86}
L.~Babai.
\newblock On {Lov\'asz'} lattice reduction and the nearest lattice point
  problem.
\newblock {\em Combinatorica}, 6(1), 1986.

\bibitem[Ban93]{BanNewBounds93}
Wojciech Banaszczyk.
\newblock New bounds in some transference theorems in the geometry of numbers.
\newblock {\em Mathematische Annalen}, 296(4), 1993.

\bibitem[DLdW19]{DLdFindingClosest19}
Emmanouil Doulgerakis, Thijs Laarhoven, and {Benne de} de~Weger.
\newblock Finding closest lattice vectors using approximate {Voronoi} cells.
\newblock In {\em PQCrypto}, 2019.
\newblock \url{https://eprint.iacr.org/2016/888}.

\bibitem[DRS14]{DRSClosestVector14}
Daniel Dadush, Oded Regev, and Noah {Stephens-Davidowitz}.
\newblock On the {Closest Vector Problem} with a distance guarantee.
\newblock In {\em CCC}, 2014.
\newblock \url{http://arxiv.org/abs/1409.8063}.

\bibitem[GN08]{GNFindingShort08}
Nicolas Gama and Phong~Q. Nguyen.
\newblock Finding short lattice vectors within {M}ordell's inequality.
\newblock In {\em STOC}, 2008.

\bibitem[HLR09]{HLRNoteDistribution09}
Ishay Haviv, Vadim Lyubashevsky, and Oded Regev.
\newblock A note on the distribution of the distance from a lattice.
\newblock {\em Discrete \& Computational Geometry}, 41(1), 2009.

\bibitem[Laa16a]{LaaFindingClosest16}
Thijs Laarhoven.
\newblock Finding closest lattice vectors using approximate {Voronoi} cells,
  2016.
\newblock \url{https://eprint.iacr.org/2016/888/20161219:141310}.

\bibitem[Laa16b]{LaaSievingClosest16}
Thijs Laarhoven.
\newblock Sieving for closest lattice vectors (with preprocessing).
\newblock In {\em SAC}, 2016.

\bibitem[LLM06]{LLMBoundedDistance06}
Yi-Kai Liu, Vadim Lyubashevsky, and Daniele Micciancio.
\newblock On {Bounded Distance Decoding} for general lattices.
\newblock In {\em RANDOM}, 2006.

\bibitem[MO90]{MOLatticePoints90}
J.~E. Mazo and A.~M. Odlyzko.
\newblock Lattice points in high-dimensional spheres.
\newblock {\em Monatshefte f\"ur Mathematik}, 110(1), 1990.

\bibitem[MP12]{MPTrapdoorsLattices12}
Daniele Micciancio and Chris Peikert.
\newblock Trapdoors for lattices: Simpler, tighter, faster, smaller.
\newblock In {\em EUROCRYPT}, 2012.
\newblock \url{https://eprint.iacr.org/2011/501}.

\bibitem[MR07]{MRWorstcaseAveragecase07}
Daniele Micciancio and Oded Regev.
\newblock Worst-case to average-case reductions based on {Gaussian} measures.
\newblock {\em SIAM Journal of Computing}, 37(1), 2007.

\bibitem[MV13]{MVDeterministicSingle13}
Daniele Micciancio and Panagiotis Voulgaris.
\newblock A deterministic single exponential time algorithm for most lattice
  problems based on {Voronoi} cell computations.
\newblock {\em SIAM Journal on Computing}, 42(3), 2013.

\bibitem[PHS19]{PHS19}
Alice {Pellet--Mary}, Guillaume Hanrot, and Damien Stehlé.
\newblock Approx-{SVP} in ideal lattices with pre-processing.
\newblock In {\em Eurocrypt}, 2019.
\newblock (to appear).

\bibitem[Sch87]{SchHierarchyPolynomial87}
Claus-Peter Schnorr.
\newblock A hierarchy of polynomial time lattice basis reduction algorithms.
\newblock {\em Theor. Comput. Sci.}, 53(23), 1987.

\bibitem[{Ste}16]{SteDiscreteGaussian16}
Noah {Stephens-Davidowitz}.
\newblock Discrete {G}aussian sampling reduces to {CVP} and {SVP}.
\newblock In {\em SODA}, 2016.
\newblock \url{http://arxiv.org/abs/1506.07490}.

\bibitem[{Ste}17]{NSDthesis}
Noah {Stephens-Davidowitz}.
\newblock {\em On the {G}aussian measure over lattices}.
\newblock {Ph.D.} thesis, New York University, 2017.

\bibitem[Ver12]{VerIntroductionNonasymptotic12}
Roman Vershynin.
\newblock Introduction to the non-asymptotic analysis of random matrices.
\newblock In {\em Compressed Sensing: Theory and Applications}. 2012.

\end{thebibliography}

\end{document}